\documentclass[11pt,draftcls,onecolumn] {IEEEtran} %TSP

\usepackage{amsmath,amssymb,graphicx,epsfig,epstopdf,fancyhdr,amsthm,tabulary}
\usepackage{cite}
\usepackage{multirow}
\usepackage{float}
\usepackage{amsfonts}
\usepackage{footnote}
\makesavenoteenv{tabular}
\usepackage{tabulary,bigstrut,array,multirow,booktabs}%
\usepackage{tikz,subfigure}
\usetikzlibrary{backgrounds,automata}
\usetikzlibrary{calc}
\usepackage{pdfsync}
\usepackage{bm}
\usepackage{latexsym}
\usepackage{flushend}
\newtheorem{lem}{Lemma}

\theoremstyle{remark}

\theoremstyle{definition}

\newcommand{\tr}{\mathop{\mathrm{tr}}}
\newcommand{\bfsub}[2]{\mbox{\boldmath$#1$}$_{#2}$}
\newcommand{\bfsubp}[3]{\mbox{\boldmath$#1$}$_{#2}^{#3}$}
\newcommand{\mybf}[1]{\mbox{\boldmath$#1$}}
\begin{document}
%\pagestyle{fancy}
%\fancyhead[LO,LE]{\sf McGill University\\Electrical and Computer Engineering}

%%%%%%%\title{On the Capacity of the Cognitive Gaussian Z-Interference Channel:\\ Less Noisy and More Capable CGZIC}
\title{Wyner-Ziv Coding in the Real Field Based on BCH-DFT Codes}

\author{\IEEEauthorblockN{Mojtaba Vaezi, {\it Student Member, IEEE,}  and Fabrice Labeau, {\it Senior Member, IEEE}}\\
%\IEEEauthorblockA{Department of Electrical and Computer Engineering\\
%McGill University\\
%Montreal, Quebec H3A 2A7, Canada\\
%Email: mojtaba.vaezi@mail.mcgill.ca, fabrice.labeau@mcgill.ca}
}


\maketitle

%----------------------------------------------------------------------
\begin{abstract}

We show how {\it real-number codes} can be used to compress correlated
sources and establish a new framework for {\it distributed lossy source
coding}, in which we quantize compressed sources
instead of compressing quantized sources. This
change in the order of {\it binning} and {\it quantization} blocks makes it
possible to model correlation between continuous-valued sources
more realistically and compensate for the quantization error when the sources
are completely correlated. We focus on the asymmetric case, i.e.,
lossy source coding with side information at the decoder, also known as {\it Wyner-Ziv coding}.
 The encoding and decoding procedures
are described in detail for discrete Fourier transform (DFT)
codes, both for {\it syndrome-} and {\it parity-based} approaches.
We also extend the parity-based approach to the case where the transmission channel is noisy
and perform distributed {\it joint source-channel coding} in this context.
The proposed system is well suited for {\it low-delay} communications.
Furthermore, the mean-squared reconstruction error (MSE)
is shown to be less than or close to the quantization error level,
the ideal case in coding based on binary codes.

\end{abstract}


\IEEEpeerreviewmaketitle

\section{Introduction}
{\let\thefootnote\relax\footnotetext{
%Manuscript received March 11, 2012; revised June 14, 2012; accepted August	
%26, 2012. Date of publication September 12, 2012; date of current version 	
%November 16, 2012. The associate editor coordinating the review of this
%manuscript and approving it for publication was Prof. Xiqi Gao.
This work was supported by Hydro-Qu\'{e}bec, the Natural Sciences and Engineering
Research Council of Canada and McGill University in the framework of the
NSERC/Hydro-Qu\'{e}bec/McGill Industrial Research Chair in Interactive
Information Infrastructure for the Power Grid.
Part of the material in this paper was presented at the
IEEE 76th Vehicular Technology Conference (VTC2012-Fall), Qu\'{e}bec city, Canada, September 2012 \cite{Vaezi2011DSC}. 	

The authors are with the Department of Electrical and Computer
Engineering, McGill University, Montreal, QC H3A 0A9, Canada (e-mail: 	
mojtaba.vaezi@mail.mcgill.ca; fabrice.labeau@mcgill.ca). 	
%Digital Object Identifier 10.1109/TSP.2011.2167617 	
}}

\IEEEPARstart{T}{he} distributed source coding (DSC) studies compression of statistically dependent
sources which do not communicate with each other \cite{SW}.
The Wyner-Ziv coding problem \cite{WZ}, a special case of {\it lossy} DSC, considers lossy
data compression with side information at the decoder.
The current approach to the DSC of continuous-valued sources is to first convert them
 to discrete-valued sources and then apply {\it lossless} (Slepian-Wolf) coder \cite{Girod,Xiong,Pradhan2}.
% In other words, the quantized source is compressed.
 Similarly, a practical Wyner-Ziv encoder consists of a quantizer and Slepian-Wolf encoder. %\cite{Pradhan2,Girod,Xiong}.
There are, hence,  {\it quantization} and {\it binning} losses for the source coder.
Despite this, rate-distortion theory promises that block codes of sufficiently
 large length are asymptotically optimal,
and they can be seen as vector quantizers followed by fixed-length coders \cite{dragotti2009distributed}.
Therefore, practical {\it Slepian-Wolf} coders have been realized using different binary channel codes,
e.g.,  LDPC and turbo codes \cite{liveris2002compression,bajcsy2001coding,aaron2002compression}.
These codes, however, are out of the question if low
delay is imposed on the system as they may introduce excessive delay when the desired probability
of error is very low.

 %e.g., in real-life sensor-network scenarios.

In this paper, we establish a new framework for the Wyner-Ziv coding
and distributed lossy source coding, in general. 
We propose to first compress the continuous-valued sources and then
quantize them, as opposed to the conventional approach.
The new framework is compared against the existing one in Fig.~\ref{fig:realDSC}.
It introduces the use of {\it real-number codes} (see, e.g., 
\cite{Marshall,wolf1983redundancy,nair1990real,marvasti1999efficient,chen2005numerically}),
to represent correlated sources with fewer samples, in the real field.
%in the compression of correlated signals.
%The compression is thus in the real field, aiming at
%representing the source with fewer samples.

To do compression, we
generate syndrome or parity samples of the input sequence using a real-number
channel code \cite{Marshall}, similar to what is done to compress a binary sequence
of data using binary channel codes. Then, we quantize these syndrome or parity
samples and transmit them. There are still coding (binning)
and quantization losses; however, since coding is performed before
quantization, error correction is in the real field, and the quantization
error can be corrected when two sources are completely correlated
over a block of code. A second and more important advantage of this
approach is the fact that the correlation channel model can be more
realistic, as it captures dependency between continuous-valued
sources rather than quantized sources.
In the conventional approach,
it is implicitly assumed that quantization of correlated signals
results in correlated sequences in the binary domain; this may
not necessarily be precise due to the nonlinearity of the quantization operation.
To avoid any loss due to the inaccuracy of correlation model, we
exploit the correlation between the continuous-valued sources before quantization.

\begin{figure}[!t]
  \centering
 \includegraphics [scale=1.2] {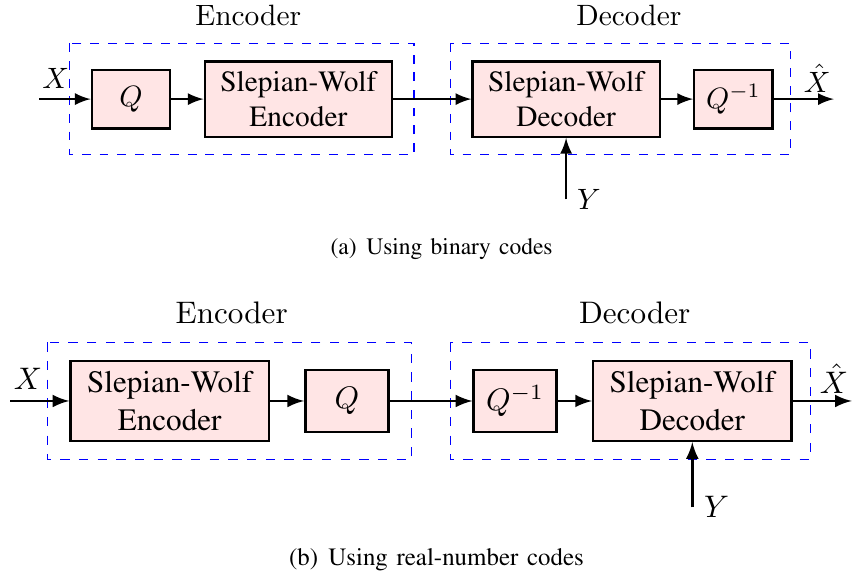}
  \caption{The Wyner-Ziv coding based on  binary and real-number codes. Both schemes can be simply extended to distributed source coding.}
\label{fig:realDSC}
\end{figure}

The new approach is also capable of alleviating the quantization error. This is possible because coding precedes quantization.
Specifically, we use the {\it Bose-Chaudhuri-Hocquenghem} (BCH) DFT codes for compression
\cite{Marshall,blahut2003algebraic,rath2004subspace,gabay2007joint,Takos,Vaezi2011LS}. Owing to
DFT codes, the loss due to quantization can be decreased by a
factor of $\frac{k}{n}$ for an $(n,k)$ code \cite{goyal2001quantized,rath2004frame,Vaezi2011LS}.
Additionally, reconstruction loss becomes zero if the two sources are perfectly correlated over one codeword.
This is achieved in view of modeling the correlation between
the two sources in the continuous domain.
DFT codes also can exploit the temporal correlation typically found in many
sources.

Finally, the proposed scheme is more suitable
for low-delay communications since, even by using short DFT codes, a reconstruction error less
than the quantization error is attainable.
Moreover, as another contribution of this paper,
we use a single DFT code both to compress and protect
sources against channel variations.
This extends the Wyner-Ziv coding to the case when errors occur during transmission
and proposes joint source-channel coding (JSCC) with side information at the decoder,
within the realm of real-number codes. This scheme directly maps
short source blocks into channel blocks, and thus it is well
suited to low-delay coding.  While the MSE performance of distributed JSCC systems with binary
codes is limited to the quantization error level, the proposed
scheme breaks through this limit.

The paper is organized as follows. In Section~\ref{sec:sys},
we motivate the new framework for lossy DSC. In
Section~\ref{sec:DFT}, we study the encoding and decoding of DFT codes, and we adapt
the subspace error localization to the Slepian-Wolf coder. Then in Section~\ref{sec:WZ}, we present the encoder
and decoder for the Wyner-Ziv coding based on DFT codes, both for the syndrome and parity approaches.
The system is extended to the noisy channel case in Section~\ref{sec:DJSCC}. Section~\ref{sec:sum}
presents the simulation results. Section~\ref{sec:con} provides our concluding remarks.

For notation, we use upper-case  and boldface lower-case letters for vectors, boldface upper-case letters for
matrices,  $(.)^T$ for transpose, $(.)^H$ for conjugate transpose,
$(.)^\ast$ for conjugate, and $\tr(.)$ for the trace.
The dimensions of matrices are indicated by one and two subscripts, when required.

%\section{Motivations and Correlation Channel Model}\label{sec:sys}
\section{Real-Number Codes for DSC}\label{sec:sys}

\subsection{Motivations}
Similar to error correction in finite fields, the basic
idea of error correcting codes in the {\it real field} is to insert
redundancy to a message vector to convert it to
a longer vector, called the codeword.
Nevertheless, the insertion of
redundancy is done in the real field, i.e.,  before quantization
and entropy coding \cite{Marshall,wolf1983redundancy}.
One main advantage of \textit{soft redundancy} (real field codes)  over \textit{hard redundancy}
(binary field codes) is that by using soft redundancy one can go beyond the
quantization error level and thus reconstruct continuous-valued
signals more accurately.
%particularly at very low and very high channel quality.
%This makes real-number codes more suitable than binary codes
%for lossy distributed source coding.

We introduce the use of real-number codes in lossy compression of correlated signals.
The proposed system is depicted in Fig. \ref{fig:realDSC}. Although it
consists of the same blocks as the existing
Wyner-Ziv coding scheme \cite{Girod, Xiong}, the order of these blocks is changed here.
Therefore binning is performed before quantization;
we use DFT codes for this purpose \cite{Marshall,rath2004subspace}.
This change in the order of binning and quantization blocks
brings some advantages to lossy DSC, as described in the following:

\subsubsection{Realistic correlation model}
In the existing framework for lossy DSC, correlation between two sources is
modeled after quantization, i.e., in the binary domain. More precisely, correlation
between quantized sources is usually modeled as a binary symmetric channel (BSC), mostly
with known crossover probability \cite{Pradhan2,Girod,Xiong,aaron2002compression}.
Admittedly though, due to nonlinearity of the quantization operation,
correlation between the quantized signals is not known accurately even if it is known in the continuous domain.
This motivates us to investigate a method that exploits correlation between continuous-valued sources
to perform DSC.

\subsubsection{Alleviating the quantization error}
In lossy data compression with side information at the decoder, soft redundancy,
added by DFT codes, can be used to correct both quantization errors and (correlation)
channel errors. Thus, the loss due to quantization can be recovered,
at least partly if not wholly. In fact,
if the two sources are exactly the same over a codeword, the quantization error
can be compensated for. That is, perfect reconstruction is attainable over the corresponding samples.
The loss due to the quantization error is decreased by a factor of code rate even
if correlation is not perfect, i.e., when (correlation) channel errors exist.
This is because DFT codes are {\it tight} frames; hence, they minimize the MSE
\cite{goyal2001quantized,rath2004frame,Vaezi2011LS,casazza2003equal}.

\subsubsection{Low-delay communications}
If communication is subject to low-delay constraints,
we cannot use turbo or LDPC codes, as their performance is not
satisfactory for short code length.
%Since, existing binary DSC codes are mainly based on turbo and  LDPC codes; DSC in binary field is not suitable if low-delay communication is required.
low-delay coding has recently drawn a lot of attention; it is
done by mapping short source blocks into channel blocks \cite{akyol2010optimal,wernersson2009nonlinear}.
%, in a linear or non-linear fashion \cite{akyol2010optimal,wernersson2009nonlinear,chen2011low}. %
Whether low-delay requirement exists or not depends on the specific
applications. However, even in the applications that low-delay
transmission is not imperative, it is sometimes useful to consider
low-dimensional systems for their low computational complexity.

\subsection{Correlation Channel Model}\label{sec:model}
Accurate modeling of the correlation between the sources plays a crucial role in
the efficiency of the DSC systems. Most of the previous works assume that the correlation between continuous-valued sources
can be modeled by an i.i.d. BSC
 in the binary field. This is not, however, exact because of the nonlinearity of quantization operation.
Besides, independent errors in the sample domain are not translated to independent errors
in the bit domain \cite{Vaezi2012model}.
These issues can be avoided if we
exploit the correlation between the continuous-valued sources before quantization.

The correlation between
the analog sources $X$ and $Y$, in general, can be defined by
\begin{align}
Y&=X+E, \label{eq:corE1}
\end{align}
where $E$ is a real-valued random variable.
Particularly, the above model represents
some well-known models motivated in video coding and sensor networks. Let
\begin{align} \label{eq:corE2}
E &  \sim \begin{cases}
{\cal N} (0,\sigma_e^2) \qquad\quad\quad\quad  \text{ w.p. } \quad q_1,\\
{\cal N} (0,\sigma_e^2+\sigma_i^2)\quad\quad  \,\,\,\,\, \text{ w.p. } \quad q_2,\\
 0 \qquad \quad\quad \quad\qquad\quad \;\, \text{ w.p. } \quad 1-q_1-q_2,\\
\end{cases}
\end{align}
in which $\sigma_i^2 \gg \sigma_e^2$ and $q_1+q_2 \leq 1$. Then,  for $q_1=1$ or $q_2=1$
the Gaussian correlation is obtained, which is broadly used in the sensor networks literature.
Further, for $q_1+q_2=1$
 the Gaussian-Bernoulli-Gaussian (GBG)  and for $q_1+q_2<1$, $q_1q_2=0$ the Gaussian-Erasure
(GE) models are realized. The latter two models are
 more suitable for video applications \cite{bassi2008source}.

\section{BCH-DFT CODES: Construction and Decoding}
\label{sec:DFT}

In this section, we study a class of real-number codes that are employed for binning throughout this paper,
investigate some properties of their syndrome, and adapt their decoding algorithm to the Slepian-Wolf coding.
These codes are a family of Bose-Chaudhuri-Hocquenghem (BCH)
codes in the real field whose parity-check matrix $\bm{H}$ and Generator matrix $\bm{G}$ are
defined based on the DFT matrix; they are, thus, known as BCH-DFT codes, or simply DFT codes.

BCH-DFT codes\cite{Marshall} are linear block codes over the {\it real} or {\it complex} fields.
Similar to other BCH codes, the spectrum of any codeword is zero in a block of $d\triangleq n-k$ cyclically adjacent
components, where $d+1$ is the designed distance of that code \cite{blahut2003algebraic}.
The error correction capability of the code is, hence, given by $t=\lfloor \frac{d}{2}\rfloor$.
In this paper, we only consider real BCH-DFT codes, i.e., the BCH-DFT codes whose generator matrix has real entries only.

\subsection{Encoding}\label{sec:Enc}

An $(n, k)$ real BCH-DFT code is defined by its generator and parity-check matrices.
The generator matrix is given by
\begin{align}
\bm{G}= \sqrt{\frac{n}{k}} \bm{W}_n^H \bm{\Sigma} \bm{W}_k,
\label{eq:G1}
\end{align}

\noindent in which $\bm{W}_k$ and $\bm{W}_n^H$ respectively are the DFT and
IDFT matrices of size $k$ and $n$, and
\mbox{\boldmath$\Sigma$} is an $n \times k$ matrix  defined as
\begin{align}
\bm{\Sigma} = \left( \begin{array}{ccccccc}
       \bm{I}_\alpha  & \bm{0}  \\
       \bm{0}   & \bm{0}  \\
      \bm{0}    &   \bm{I}_\beta    \\
      \end{array}
\right), \label{eq:cov}
\end{align}
where $\alpha = \lceil \frac{n}{2}\rceil  -\lfloor \frac{n-k}{2}\rfloor$\footnote{Knowing that
$n$ and $k$ cannot be simultaneously even for a real DFT code \cite{Marshall}, one can show that $\alpha =  \lceil \frac{k+1}{2}\rceil$.}, $\beta=k-\alpha$,
and the sizes of zero blocks are such that $\Sigma$ is an $n\times k$ matrix \cite{vaezi2012frame,rath2004subspace,Takos,gabay2007joint}.
Then, for any $\bm{u}$, this enforces the spectrum of the codeword
\begin{align}
\bm{c}=\bm{G}\bm{u},
\label{eq:code}
\end{align}
to have $n-k$ consecutive zeros, which is
required for any BCH code \cite{blahut2003algebraic}.
The parity-check matrix $\bm{H}$, on the other hand, is constructed by using the $n-k$ columns of $\bm{W}_n^H$
corresponding to the $n-k$ zero rows of $\bm{\Sigma}$.
Therefore, by virtue of the unitary property of $\bm{W}_n^H$, $\bm{H}$ is the null space of $\bm{G}$, i.e.,
\begin{align}
\bm{HG}=\bm{0}.
\label{eq:HG}
\end{align}

In the rest of this paper, we use the term DFT code in lieu of real BCH-DFT code.

\subsection{Decoding}\label{sec:Dec}

Before introducing the decoding algorithm, we define some notation and basic concepts.
Let $\bm{r} = \bm{c} + \bm{e}$ be the received vector (a noisy version of $\bm{c}$),
where $\bm{c}$ is a codeword generated by \eqref{eq:code}. Suppose that $\bm{e}$ is an error vector
with $\nu$ nonzero elements at positions $ i_1, \hdots, i_{\nu}$; the magnitude of error at position $i_p$ is $e_{i_p}$.
 Then, we can compute
\begin{align}
\mbox{\boldmath$ \bm{s}= \bm{Hr}= \bm{H}(\bm{c} + \bm{e})= \bm{He}$},
\label{eq:synd}
\end{align}
where $\bm{s}=[s_1,\, s_2, \hdots, s_d]^T$  is a complex vector with
\begin{align}
s_m=\frac{1}{\sqrt{n}}\sum_{p=1}^{\nu} e_{i_p}X_p^{\alpha-1+m}, \qquad m=1,\hdots, d
\label{eq:syndsamp}
\end{align}
in which $\alpha =  \lceil \frac{k+1}{2}\rceil$ as defined in \eqref{eq:cov}, $X_p= e^{ \frac{-j2\pi i_p}{n}}$, and  $ p=1, \hdots, \nu.$
Next, we define the syndrome matrix
\begin{align}
\bm{S}_m&=  \left[ \begin{array}{cccc}
 s_1 & s_2 & \hdots &s_{d-m+1}  \\
 s_2 & s_3 & \hdots & s_{d-m+2}  \\
 \vdots & \vdots & \ddots & \vdots\\
s_{m} & s_{m+1} & \hdots & s_{d}  \\  \end{array}
\right],
\label{eq:syndmatrix}
 \end{align}
for $\nu+1 \le m \le d-\nu+1$ \cite{rath2004subspace}. Also, we define the covariance matrix as
\begin{align}
\bm{R}=\bm{S}_m\bm{S}_m^H.
\label{eq:R}
\end{align}

For decoding, we use the extension of the well-known Peterson-Gorenstein-Zierler
(PGZ) algorithm to the real field \cite{blahut2003algebraic}.
This algorithm, aimed at detecting, localizing, and calculating the errors,
works based on the syndrome of error. We summarize the main steps of this algorithm, adapted for a DFT code of length $n$, in the following.
\begin{itemize}
  \item \textbf{Error detection:} Determine the number of errors $\nu$ by constructing a syndrome matrix and finding its rank.
  \item \textbf{Error localization:} Find the coefficients $\Lambda_1, \hdots, \Lambda_\nu$ of the {\it error-locating polynomial} $\Lambda(x) = \prod_{\substack{i=1}}^{\nu} (1-xX_i)$ whose roots
   $X_1^{-1}, \hdots, X_\nu^{-1}$ are used to determine error locations; the errors are then
in the locations $i_1, \hdots, i_\nu$ such that $X_1 = \omega^{i_1}, \hdots, X_\nu = \omega^{i_\nu}$ and  $\omega= e^{-j\frac{2\pi}{n}}$.
  \item \textbf{Error calculation:} Finally, calculate the error magnitudes by solving a set of linear equations whose constants coefficients are powers of $X_i$.
\end{itemize}

In practice however, the received vector is distorted because of quantization.
Let $\bm{\hat c}$ and $\bm{q}$ denote the quantized codeword and quantization
noise so that $\bm{\hat c}=\bm{c}+\bm{q}$. Therefore, $\bm{r} = \bm{\hat c} + \bm{e}$ and
 its syndrome is no longer equal to the syndrome of error because
\begin{align}
\bm{Hr}= \bm{H}(\bm{ c} + \bm{q}+ \bm{e})= \bm{s}_q + \bm{s} = \tilde{\bm{s}},
\label{eq:syndq}
\end{align}
where $\bm{s}_q \equiv\bm{Hq}$ and  $\bm{q}=[q_1,\, q_2, \hdots, q_n]^T$ is the quantization error.
The distorted syndrome samples can be written as
\begin{align}
\tilde{s}_m = \frac{1}{\sqrt{n}}\sum_{p=1}^{\nu} e_{i_p}X_p^{\alpha -1 + m} +  \frac{1}{\sqrt{n}}\sum_{p'=1}^{n} q_{i_{p'}}X_{p'}^{p'-1}.
\end{align}
The distorted syndrome matrix $\tilde{\bm{S}}_m$  and the corresponding covariance matrix $\tilde{\bm{R}}=\tilde{\bm{S}}_m\tilde{\bm{S}}_m^H$
are defined similar to \eqref{eq:syndmatrix} and \eqref{eq:R} but for the distorted syndrome samples.

While the exact value of the error is determined neglecting quantization,
the decoding becomes an {\it estimation} problem in the presence
of quantization. Then, it is imperative to modify the PGZ algorithm to decode the
errors reliably \cite{blahut2003algebraic,rath2004subspace,Takos,gabay2007joint,Vaezi2011LS}.
In the remainder of this section, we discuss this problem and also improve the error
detection and localization, by introducing a slightly different version
of the existing methods.

\subsubsection{Error detection}
\label{subsec:Pd}
For a given DFT code, we first fix an empirical threshold $\theta$
 based on eigendecomposition of $\tilde{\bm{R}}$ when the codewords are
 error-free, i.e., when only the quantization error exist.
This threshold is on the magnitude of eigenvalues, rather than the determinant of $\tilde{\bm{R}}$.
Let $\lambda_{\mathrm{max}}$ denote the largest eigenvalue of $\tilde{\bm{R}}$ for $m=t+1$.
We find $\theta$ such that, for a desired probability of correct detection $p_d$,
\begin{align}
\mathrm{Pr}(\lambda_{\mathrm{max}}<\theta)\ge p_d.
\end{align}
In practice, when errors can occur, we estimate the number of errors
by the number of eigenvalues of $\tilde{\bm{R}}$
greater than $\theta$. This one step estimation is better than the original
estimation in the PGZ algorithm \cite{blahut2003algebraic,Takos},
where the last row and column of $\bm{S}_t$ are removed until
coming up with a non-singular matrix. The improvement comes from
incorporating all syndrome samples, rather than some of them, for decision making.
\subsubsection{Error localization} \label{sec:DecLoc}
The {\it Subspace}  or {\it coding-theoretic} error localizations \cite{rath2004subspace}
can be used to find the coefficients $\Lambda_1, \hdots, \Lambda_\nu$ of the error-locating polynomial.
The subspace approach is, however, more general than the coding-theoretic approach in the sense that
it can use up to $t+1-\nu$ degrees of freedom to localize $\nu$ errors,
compared to just one degree of freedom in the coding-theoretic approach.
Hence, it improves the error localization in the presence of the quantization noise.
%and yields better results for $\nu < t$ \cite{rath2004subspace}.
We apply the subspace error localization  both to the syndrome-
and parity-based DSC (see Section~\ref{sec:WZ}),
similar to that in channel coding \cite{rath2004subspace}.
To this end, we eigen-decompose the covariance matrix
$\tilde{\bm{R}}=\tilde{\bm{S}}_m\tilde{\bm{S}}_m^H$ for $m=t+1.$\footnote
{Although $\nu+1 \le m \le d-\nu+1$, the best result is achieved for $m=t+1$ \cite{rath2004subspace}.
For this $m$, the size of $\tilde{\bm{S}}_m$ is either $(t+1)\times (t+1)$ or $(t+1) \times t$.}
It results in two orthogonal subspaces, the {\it error} and  {\it noise} subspaces.
There are $m-\nu$ vectors in the noise subspace; we use them to localize errors
using the MUSIC algorithm \cite{rath2004subspace,schmidt1986multiple}.
However, one should note that the way we compute the syndrome
in DSC is different from that in channel coding which was presented in \eqref{eq:synd}
and \eqref{eq:syndq}; this will be elaborated in Section~\ref{sec:WZ}.

\subsubsection{Error calculation}
This last step is rather simple. Let $\bm{H}_e$ be
the matrix consisting of the columns of \mbox{\boldmath$H$} corresponding to error indices.
The errors magnitude
 $\bm{e}= \left[ e_{i_1},\, e_{i_2}, \hdots, e_{i_\nu} \right]^T $
can be determined by solving
\begin{align}
%\mbox{\boldmath$ {H}_e Y ={s}_e$},
 \bm{H}_{e}\bm{e}=\tilde{\bm{s}},
  \label{eq:dec1}
 \end{align}
in a {\it least squares} sense, for example.
This completes the error correction algorithm by calculating the error
 vector.

\subsection{Performance Compared to Binary Codes}\label{sec:Per}

DFT codes by construction are capable of decreasing the quantization error.
When there is no error, an $(n, k)$  DFT code brings down the MSE
below the quantization error level with a factor of $R_c=k/n$ \cite{rath2004frame,goyal2001quantized}.
%This is also shown to be valid for channel errors, as long as no error correction is
%carried out \cite{Vaezi2011LS}. Likewise, one can see the above argument about the MSE between
%information samples and their reconstruction is valid when error correction is
%performed.
This is also shown to be valid for channel errors, as long as channel can be modeled as an additive noise.
To appreciate this, one can consider the generator matrix of a DFT code as
{\it analysis frame operator} of a tight frame \cite{rath2004frame}; it is
known that frames are resilient to any additive noise, and
tight frames reduce the MSE $k/n$ times \cite{kovacevic2008introduction}.
Hence, DFT codes can result in
a MSE even less than the quantization error level whereas the
 MSE in a binary code is obviously lower-bounded by the quantization error level.

%%%%%%%%%%%%%%%%%%%%%%%%%%%%%%%%%%%%%%%%%%%%%%%%%%%%%%%%%%%%%%%%%%%%%%%%%%%%%%%%%%%

\section{Wyner-Ziv Coding Using DFT Codes} \label{sec:WZ}

 In this section, we use
DFT codes to do Wyner-Ziv coding in the real field.
This is accomplished by using DFT codes for binning and
transmitting compressed signal, in the form of syndrome or parity samples,
in a digital communication system.
Let $\mybf{x}$ be a sequence of real random variables $x_1x_2 \hdots x_n$,
and $\mybf{y}$ be a noisy version of $\mybf{x}$
such that $y_i=x_i + e_i$, where $e_i$ is continuous, i.i.d., and independent of $x_i$, as described in \eqref{eq:corE2}.
The lower-case letters $x$, $y$, and $e$, respectively, are used to show the realization of the random variables $X$, $Y$, and $E$.
Since $\bm{e}$ is continuous, this model precisely captures any variation of $\bm{x}$, so it can
model correlation between $\mybf{x}$ and $\mybf{y}$ accurately.
This correlation model is important, for example, in video coders
that exploit Wyner-Ziv concepts, e.g., when the decoder builds side information
via extrapolation of previously decoded frames or interpolation of key frames \cite{bassi2008source}.
In this paper, we use the GE model with $q_2=0$ in \eqref{eq:corE2} where $\bm{e}$ is a sparse vector
inserting up to $t$ errors in each codeword.

\subsection{Syndrome Approach} \label{synd}
\subsubsection{Encoding}

\begin{figure}[!t]
  \centering
 \includegraphics [scale=1.2] {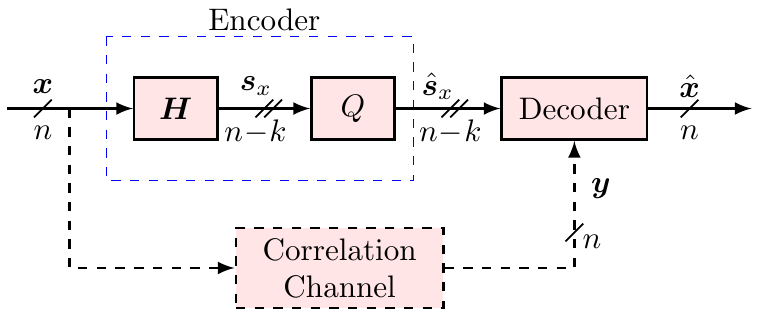}
  \caption{The Wyner-Ziv coding using DFT codes: Syndrome approach.}
\label{fig:WZsynd}
\end{figure}

Given \mybf{H}, to compress an arbitrary sequence of data samples, we
multiply it with \mybf{H} to find the corresponding
syndrome samples \bfsub{s}{x}\mybf{=Hx}. The syndrome
is then quantized (\bfsubp{\hat s}{x}{}$\;=\;$\bfsubp{s}{x}{}$\;+\; \mybf{q}$) and transmitted
over a noiseless digital communication system, as shown in
Fig. \ref{fig:WZsynd}. Note that
 \mbox{\boldmath$s$}$_x$, \bfsub{\hat s}{x} are both complex vectors of length $n-k$.
Thus, it seems that to transmit each sample we need to send two real numbers,
one for the real part and one for the imaginary part, which halves the compression ratio.
However, we observe that the syndrome of a DFT code is symmetric, as stated below.
%as stated in the following lemma.

\begin{lem}
The syndrome of an $(n,k)$ DFT code satisfies
%the following statements hold:\\
%i. For an odd $k$ we have $s_{d-m+1}=s_m^{\ast}$\\
%ii. For an even $k$ we have $s_{d-m}=s_m^{\ast}$
\begin{align}
s_m=  \left\{
  \begin{array}{l l}
    s_{d-m+1}^{\ast}, & \quad \text{if $k$ is odd,}\\
     s_{d-m}^{\ast}, & \quad \text{if $k$ is even,}
     \label{consts}
  \end{array} \right.
  \end{align}
for $m=1,\hdots,d$ and  $d\triangleq n-k$.
\label{lem1}
\end{lem}

\begin{proof}
See Appendix~\ref{sec:app1}.
\end{proof}
The above lemma implies that, for any $d$, it suffices to know the first $\lceil \frac{d}{2} \rceil$ syndrome samples.
 We know that syndromes are complex numbers in general; however, from $s_{d-m+1}=s_m^{\ast}$
it is clear that if $d$ is an odd number, $s_m$ is real for $m=\lceil \frac{d}{2} \rceil$.
Therefore, for an $(n,k)$ code with odd $k$, transmitting $n-k$ real numbers
suffices. This results in a compression ratio of $\frac{n}{n-k}$ for the Slepian-Wolf encoder.
Yet, one can check that for even $k$, we have to transmit $n-k+1$ real samples, which incurs a slight
loss in compression. It is however negligible for large $n$.

\subsubsection{Decoding}

The decoder estimates the input sequence from the received syndrome and side
information $\bm{y}$. To this end,
it needs to evaluate the syndrome of (correlation) channel errors. This
can be simply done by subtracting the received syndrome from the syndrome of
the side information. Then, neglecting the quantization error, we obtain,
\begin{align}
\mbox{\bfsubp{s}{e}{}$ \; = \; $\bfsubp{s}{y}{}$\;- \;$\bfsubp{s}{x}{}},
\label{eq:synd5}
\end{align}
and \bfsubp{s}{e}{} can be used to precisely estimate the error vector,
as described in Section \ref{sec:Dec}.
In practice, however, the decoder knows \bfsubp{\hat s}{x}{}$\;=\;$\bfsubp{s}{x}{}$\;+\; \mybf{q}$
rather than \mbox{\boldmath$ s_x$}. Therefore, only a distorted syndrome
of error is available, i.e.,
\begin{align}
\mbox{\bfsubp{\tilde{s}}{e}{}$ \; = \; $\bfsubp{s}{y}{}$ \; - \;$\bfsubp{\hat s}{x}{}$ \; = \; $\bfsubp{s}{e}{}$ \; - \; $ \mybf{q}}.
\label{eq:synd6}
\end{align}
Hence, using the PGZ algorithm, error correction is accomplished based on \eqref{eq:synd6}.
Note that, having computed the syndrome of error, decoding algorithm
in a DSC using DFT codes is exactly the same as that in the channel coding problem.
This is different from DSC techniques in the binary field which
usually require a slight modification in the corresponding
channel coding algorithm to be customized for  DSC.

\subsection{Parity Approach} \label{par}

The syndrome-based Wyner-Ziv coding is straightforward, but it is not clear how
we can use it for noisy transmission.
In the sequel, we explore an alternative approached, namely parity-based approach, to the Wyner-Ziv coding.

\subsubsection{Encoding}
\label{subsec:parEnc}
To compress $\bm{x}$, the encoder generates the corresponding
parity sequence $\bm{p}$ with $n-k$
samples. The parity is
then quantized and transmitted, as shown in
Fig.~\ref{fig:WZparity}, instead of transmitting the input data.
To this end, we need to find a systematic
generator matrix $\bm{G}_{\mathrm{sys}}$, as $\bm{G}$ in \eqref{eq:G1}
is not in the systematic form.
%Let $\bm{H}$ be partitioned as $\bm{H} = [\bm{H}_{n-k\times k} \;|\; \bm{H}_{n-k}]$.
%Since $\bm{H}_{n-k}$ is  a Vandermonde matrix, $\bm{H}_{n-k}^{-1}$
%exist and we can write
%\begin{align}
%\bm{H}_{\mathrm{sys}}= \bm{H}_{n-k}^{-1} \bm{H}= [-\bm{P} \;|\; \bm{I}_{n-k}],
%\label{eq:synd5}
%\end{align}
%in which $\bm{P}= - \bm{H}_{n-k}^{-1}\bm{H}_{n-k \times k}$ is an $(n-k)\times k$ matrix, and
%$\bm{I}_{n-k}$ is an identity matrix of size $n-k$.
%
%The systematic generator matrix corresponding to $ \bm{H}_{\mathrm{sys}}$
%is given by
%\begin{align}
%\bm{G}_{\mathrm{sys}} &= \left[\begin{array}{c}
%      \bm{I}_k \\ %\hline %\vspace {.4 cm}
%      \bm{P}
%    \end{array}\right].
%   % =\left[\begin{array}{c}
%%      \bm{I}_k \\ %\hline %\vspace {.4 cm}
%%      \bm{H}_{n-k}^{-1}\bm{H}_{n-k \times k}
%%    \end{array}\right].
%\label{eq:Gsys}
%\end{align}
%Clearly, $\bm{H}_{\mathrm{sys}}\bm{G}_{\mathrm{sys}}=\bm{0}$. It is also easy to check that
%\begin{align}
% \bm{H}\bm{G}_{\mathrm{sys}}=\bm{0}.
%\label{eq:null}
%\end{align}
%Therefore, we do not need to calculate $\bm{H}_{\mathrm{sys}}$ and the same
%parity-check matrix $\bm{H}$ can be used for decoding in the
%parity approach.

A first approach is to find  $\bm{H}_{\mathrm{sys}}$ and build $\bm{G}_{\mathrm{sys}}$ based on that \cite{Vaezi2011DSC}.
Another, simpler, way is to obtain a systematic generator matrix directly from $\bm{G}$.
Let $\bm{G}_{k}$ be a square matrix of size $k$ composed of arbitrary rows of $\bm{G}$.
We see that $\bm{G}_{k}$ is invertible because using \eqref{eq:G1} any $k\times k$ submatrix of
$\bm{G}$ can be represented as product of a Vandermonde matrix and the DFT matrix $\bm{W}_{k}$.
This is also proven using a different approach in \cite{rath2004frame}, where it is shown that any subframe of $\bm{G}$
is a frame, and its rank is equal to $k$.
Hence, a systematic generator matrix is given by
  \begin{align}
  \bm{G}_{\mathrm{sys}} =\bm{G}\bm{G}_{k}^{-1}.
\label{eq:Gsys2}
\end{align}
Besides, from $\bm{H}\bm{G}=\bm{0}$, it is clear that
\begin{align}
 \bm{H}\bm{G}_{\mathrm{sys}}=\bm{0}.
\label{eq:null}
\end{align}
Therefore, we do not need to calculate $\bm{H}_{\mathrm{sys}}$, and the same
parity-check matrix $\bm{H}$ can be used for decoding in the
parity approach.
It is also obvious that $\bm{G}_{\mathrm{sys}}$ is a real matrix.
The question that
remains to be answered is whether $\bm{G}_{\mathrm{sys}}$ corresponds to a BCH code.
To generate a BCH code, $\bm{G}_{\mathrm{sys}}$ must have $n-k$ consecutive
zeros in the transform domain. The Fourier transform of this matrix $\bm{W}_n\bm{G}_{\mathrm{sys}}=(\bm{W}_n\bm{G})\bm{G}_{1}^{-1}$
 satisfies the required condition because $\bm{W}_n\bm{G}$,
the Fourier transform of original matrix, satisfies that.

\begin{figure}[!t]
  \centering
 \includegraphics [scale=1.2] {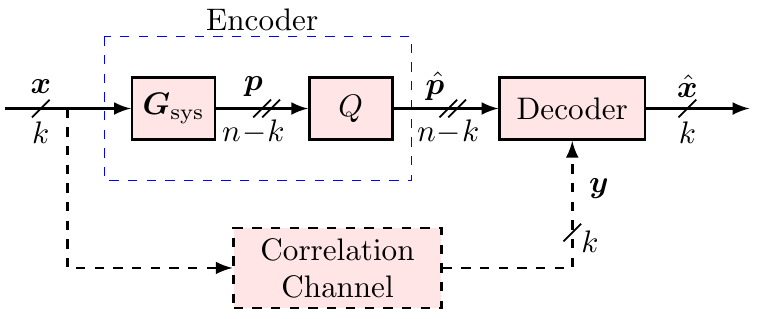}
  \caption{The Wyner-Ziv coding using DFT codes: Parity approach.}
\label{fig:WZparity}
\end{figure}

It should be emphasized that one can arbitrarily choose the rows of $\bm{G}_{k}$ in \eqref{eq:Gsys2};
this results in $\binom{n}{k}$ systematic generator matrix for an $(n,k)$ DFT code.
Although any of those systematic codes can be used for encoding,
the dynamic range of the generated parity samples depends on their relative position
of the chosen rows \cite{vaezi2012frame}.
In \cite[Theorem~7]{vaezi2012frameJ}, we have proved  that when using these
systematic frames for error correction, the mean-squared reconstruction error
 is minimized when the systematic rows are chosen as evenly spaced as possible.
 %This implies  $n- \lfloor \frac{n}{k} \rfloor k$ systematic rows with successive circular distance $\lceil \frac{n}{k} \rceil $.
 In the extreme scenario, where the systematic rows are equally spaced, the
 systematic frame is also {\it tight}. This is realized only when $n$ is an integer multiple of $k$.
 %It is well known that tight frames minimize reconstruction error \cite{goyal2001quantized,kovacevic2007life,rath2004frame,casazza2003equal}.
 Such a frame lends itself well to minimize reconstruction error \cite{goyal2001quantized,kovacevic2007life,rath2004frame,casazza2003equal}.

Finally, seeing that parity samples are real numbers,
using an $(n, k)$ DFT code, a compression ratio  of $\frac{k}{n-k}$
is achieved. Obviously, a compression ratio  of $\frac{n}{n-k}$ is
achievable if we use a $(2n-k, n)$ DFT code.

\subsubsection{Decoding}
\label{subsec:parDec}
A parity decoder estimates the input sequence from the received parity and side information $\bm{y}$.
Similar to the syndrome approach, at the decoder, we need to find the syndrome
of (correlation) channel errors.
To do so, we append the parity to the side information and
form a vector of length $n$ whose syndrome, neglecting quantization, is equal to the
syndrome of error. That is,
 \begin{align}
\bm{z}=\left[ \begin{array}{cc} \bm{y}   \\
 \bm{p} \\  \end{array} \right]
 =\left[ \begin{array}{cc} \bm{x}   \\
 \bm{p}  \end{array}\right]  +\left[ \begin{array}{cc} \bm{e}   \\
 \bm{0}\end{array}\right]= \bm{G}_{\mathrm{sys}}\bm{x}+\bm{e}',
 \end{align}
and $\bm{e}' = [\bm{e} \;|\; \bm{0}]^T $. Hence,
\begin{align}
 \bm{s}_z=\bm{s}_{e'}.
 \end{align}
Similarly, when quantization is involved ($\bm{\hat p}=\bm{p}+\bm{q}$), we get
  \begin{align}
 \bm{\tilde{z}} =\left[ \begin{array}{cc} \bm{y}   \\
 \bm{\hat p} \\  \end{array} \right]
 =\bm{z} +  \left[ \begin{array}{cc} \bm{0}   \\
\bm{q}\end{array}\right]= \bm{G}_{\mathrm{sys}}\bm{x}+\bm{e}'+\bm{q}',
 \end{align}
and
\begin{align}
 \bm{s}_{\tilde{z}}=\bm{s}_{e'}+\bm{s}_{q'},
 \end{align}
 where, $\bm{q}' = [\bm{q} \;|\; \bm{0}]^T $, and  $\bm{s}_{q'} \equiv \bm{Hq}'$.
Therefore, we obtain a distorted version of error syndrome.
In both cases, the rest of the algorithm, which is based on the
syndrome of error, is similar to that in the channel coding problem using DFT codes,
as explained in Section~\ref{sec:DecLoc}.

Error localization algorithm for the parity-based DSC \cite{Vaezi2011DSC} can be further improved
using the fact that parity samples are error-free.
As parity samples are transmitted over a noiseless channel, the error locations, in the codewords,
are restricted to the systematic samples.
Therefore, we can exclude the set of roots corresponding to the location of the parity samples.
We call this {\it adapted} error localization.
Furthermore, it makes sense to use a code with evenly-spaced parity samples so as to
optimize the location of error-free and error-prone samples in the codewords.
Such a code maximizes the distance between the
error-prone roots of the code; hence, it helps decrease the probability of incorrect decision.

\subsection{Comparison Between the Two Approaches} \label{comp}

\subsubsection{Rate}
As it was shown earlier, using an $(n, k)$ code
the compression ratio in the syndrome and parity approaches, respectively, is $\frac{n}{n-k}$
and $\frac{k}{n-k}$. Hence, for a given code, the parity approach is $\frac{k}{n}=R_c <1$ times
less efficient than the syndrome approach. Conversely, we can find two different codes that result in
same compression ratio, say $\frac{n}{n-k}$. We mentioned that in the parity approach,
a $(2n-k, n)$ code can be used for this matter, whereas an $(n, k)$
DFT code gives the desired compression ratio in the syndrome approach.
Thus, for a given compression ratio the syndrome approach implies a code with smaller rate compared to the code
required in the parity approach.

\subsubsection{Performance}
From frame theory, we know that DFT frames are tight, and
an $(n,k)$ tight frame reduces the quantization error with a factor of $R_c=\frac{k}{n}$ \cite{kovacevic2008introduction,rath2004frame,goyal2001quantized}.
This result is extended to errors, given that channel can be modeled by an additive noise \cite{Vaezi2011LS}.
The MSE performance of systematic DFT frames also linearly depends on the code rate, though they are not necessarily tight \cite{vaezi2012frame,vaezi2012frameJ}.
Therefore, for codes with the same error correction capability, the lower the code rate the better the error correction performance.
This implies a better performance for syndrome-based DSC.
Further, a $(2n-k, n)$ code has $n-k$ roots more than an $(n, k)$ code on the unit circle; hence,
the roots are closer to each other and the probability of incorrect localization of errors increases.
%On the other hand, owing to the improved subspace localization, introduced for the
%parity-based DSC in Section~\ref{subsec:parDec}, we may expect a better error localization for this approach.

Additionally, from rate-distortion theory we know that the rate required to transmit a
Gaussian source logarithmically increases with the source variance \cite{Cover}.
Thus, in a system that uses a real-number code for encoding, since coding is performed before quantization,
 the variance of transmitted sequence depends on the behavior of the
encoding matrix. In the syndrome-based DSC we transmit $\bm{s}=\bm{H}\bm{x}$.
One can check that $\sigma_{\bm{s}}=\sigma_{\bm{x}}$ \cite{vaezi2012frameJ}.
Unlike that, in the parity-based DSC, the variance of the parity samples is larger than that of the inputs.
More precisely, in an $(n,k)$  systematic DFT code, if  $\bm{c}=\bm{G}_{\mathrm{sys}}\bm{x}$, then $\sigma^2_{\bm{c}}=\gamma  \sigma^2_{\bm{x}}$ where
$\gamma = \frac{1}{n}\tr\left(\bm{G}_{\mathrm{sys}}^H\bm{G}_{\mathrm{sys}}\right) \ge 1$ \cite{vaezi2012frame}.
Since we can write  $c=[\bm{x} \;|\; \bm{p}]^T$, we have
\begin{align}
\sigma^2_{\bm{p}}= \frac{\gamma n -k}{n-k} \sigma^2_{\bm{x}} \ge \sigma^2_{\bm{x}}.
\label{eq:parVar}
\end{align}
From \cite[Theorem 7]{vaezi2012frameJ}, we know that the smallest  $\sigma_{\bm{p}}$
for a given DFT code is achieved when the parity samples, in the corresponding codewords, are located as ``evenly" as possible.
%We use such optimal systematic codes in this paper.

%Considering the above arguments, it is hard to say which approach is better; we leave this to Section \eqref{sec:sum}.
Considering the above arguments, one may expect the syndrome-based
approach to perform better than the parity-based one, for a given code or fixed compression ratio.
This is verified numerically in Section \ref{sec:sum}.
The parity-based DSC, however, has other advantages. For example, by puncturing some
parity samples {\it rate-adaptive} DSC, in the real field, is realized. Besides,
it can be easily extended to distributed joint source-channel coding, as explained
in the following section.

\section{Distributed Joint Source and Channel Coding}
\label{sec:DJSCC}

The concept of lossy DSC and Wyner-Ziv coding using DFT codes
was explained both for the syndrome and parity approaches in Section~\ref{sec:WZ}, where
syndrome or parity samples are quantized and transmitted over a {\it noiseless} channel.
This implies {\it separate} source and channel coding.
Although simple, the {\it separation theorem} is based
on several assumptions, such as the source and channel coders not being
constrained in terms of complexity and delay, which do not hold in many situations. It breaks
down, for example, for non-ergodic channels and real-time communication. In such
cases, it makes sense to integrate the design of the source and channel
coder systems, because {\it joint source-channel coding} (JSCC) can perform better
given a fixed complexity and/or delay constraints. Likewise, distributed JSCC (DJSCC) has been shown to
outperform separate distributed source and channel coding in some practical cases \cite{xu2007distributed}.
DJSCC has been addressed in \cite{aaron2002compression,xu2007distributed,liveris2002joint,garcia2005distributed},
using different binary codes.

In this section, we extent the parity-based Wyner-Ziv coding of analog sources to the
case where errors in the transmission are allowed. Thus, we introduce distributed JSSC
of analog correlated sources in the analog domain.
To do this, we use a single DFT code both to {\it compress} $\bm{x}$ and {\it protect} it
against channel variations; this gives rise to a new framework for DJSCC,
in which quantization is performed after
doing JSCC in the analog domain.
This scheme directly maps short source
blocks into channel blocks, and thus it is well suited to low-delay coding.

\begin{figure}[!t]
  \centering
 \includegraphics [scale=1.2] {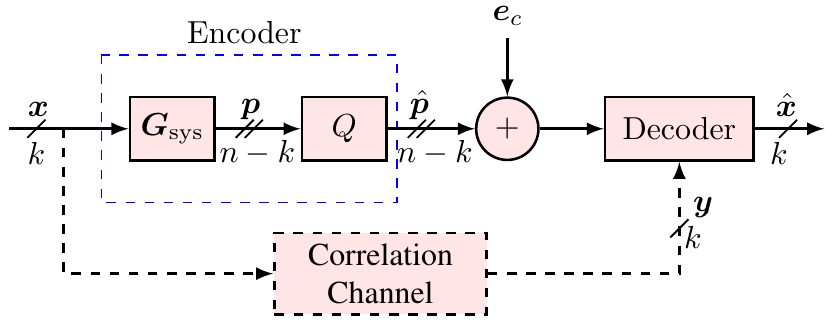}
  \caption{Joint source-channel coding (JSCC) with side information at the decoder based on DFT codes.
Obviously, the scheme can be extended to distributed JSCC.}
\label{fig:DJSCC}
\end{figure}

\subsection{Coding and Compression}

To compress and protect $\bm{x}$, the encoder generates the
parity sequence $\bm{p}$ of $n-k$
samples, with respect to a good systematic DFT code.
The parity is then quantized and transmitted over a noisy channel, as shown in
Fig.~\ref{fig:DJSCC}. To keep the dynamic range of parity samples as small as possible,
we make use of optimal systematic DFT codes, proposed in  \cite{vaezi2012frameJ}.
This increases the efficiency of the system for a fixed number of quantization levels.
Using an $(n, k)$ DFT code a total compression ration of $\frac{n}{n-k}$
is achieved. Obviously, if $n<2k$ compression is possible. However, since there
is little redundancy the end-to-end distortion could be high. Conversely, a code with $n>2k$
expands input sequence by adding \textit{soft redundancy} to protect it in a noisy channel.

\subsection{Decoding}

Let $\tilde{\bm{p}}=\hat{\bm{p}}+ \bm{e}_c$ be the received parity vector
which is distorted by quantization error $\bm{q}$ ($\hat{\bm{p}}= \bm{p}+\bm{q}$) and channel error $\bm{e}_c$.
Also, let $\bm{y} = \bm{x} + \bm{e}_v$ denote side information where $\bm{e}_v$ represents
the error due to the ``virtual'' correlation channel.
The objective of the decoder is to estimate the input sequence
from the received parity and side information.
Although we only need to determine $\bm{e}_v$, effectively it is required to find both $\bm{e}_v$
and $\bm{e}_c$.
From an error correction point of view, this is equal to finding the error vector $\bm{e} = [ \bm{e}_v \;\; \bm{e}_c]^T$
that affects the codeword $ [ \bm{x} \;\; \bm{p}]^T$.
Hence, to find the syndrome of error at the decoder, we append the parity $\tilde{\bm{p}}$ to the side information $\bm{y}$ and
form $\tilde{\bm{z}}$, a valid codeword
perturbed by quantization and channel errors,
\begin{align}
\tilde{\bm{z}}
% =\left[ \begin{array}{cc} \bm{y}   \\  \bm{\tilde{ p}} \\  \end{array} \right]
 =\left[ \begin{array}{cc} \bm{x}   \\
 \bm{p}  \end{array}\right]  + \left[ \begin{array}{cc} \bm{e}_v   \\
 \bm{e}_c \end{array}\right] +  \left[ \begin{array}{cc} \bm{0}   \\
\bm{q}\end{array}\right]= \bm{G}_{\mathrm{sys}}\bm{x}+\bm{e}+\bm{q}'.
 \end{align}
Multiplying both sides by $\bm{H}$, we obtain
\begin{align}
 \bm{s}_{\tilde{z}}=\bm{s}_e+\bm{s}_{q'},
  \end{align}
where  $\bm{s}_e \equiv \bm{He}$ and $\bm{s}_{q'} \equiv \bm{Hq}'$.
Again, we use the GE model with $q_2=0$ in \eqref{eq:corE2} to generate $\bm{e}$.
It should be emphasized that for $\bm{q}=\bm{0}$, error vector
 can be determined exactly, as long as the number of errors is not greater than $t$.
In practice, quantization is also involved, and we
obtain only a distorted version of error syndrome.
Knowing the syndrome of error, we use the error detection and
localization algorithm, explained in Section~\ref{sec:Dec}, to find and correct error.

%
%Although the extension of parity-based DSC to DJSCC is straightforward,
%it is not clear how to do this for syndrome-based DSC. This is because, in a syndrome-based DSC with noisy transmission,
%the decoder can only form $\bm{s}_{e_v}+\bm{e}_c$, where $\bm{s}_{e_v}$ is
%the difference between the transmitted syndrome and syndrome of side information, i.e.,
%$\bm{s}_{e_v}= \bm{s}_{y}-\bm{s}_{x}$, as it was in the DSC \cite{Vaezi2011DSC}. However,
%with $\bm{s}_{e_v}+\bm{e}_c$ the rank of the syndrome matrix $\bm{S}_t$ is not necessarily
%equal to $\nu$, even if quantization error is assumed to be zero. Therefore, the PGZ and subspace methods fail to find the number and location of errors.

\section{Simulation Results}
\label{sec:sum}

We evaluate the performance of the proposed systems using
a Gauss-Markov source with mean zero, variance one, and
correlation coefficient 0.9; the effective range of the input
sequences is thus $[-4, 4]$.
The sources are binned using a $(7, 5)$ DFT code.
The compressed vector, either syndrome or parity, is then
quantized with a 6-bit uniform quantizer with step size $\Delta= 0.125$ and
transmitted over a noiseless communication media.
The ``virtual" correlation channel inserts $t$ errors,
generated by $\mathcal N(0,\sigma_e^2)$.
The decoder detects, localizes, and decodes errors.
We compare the MSE between transmitted and reconstructed
codewords, to measurers end-to-end distortion.
In all simulations, we use $10^5$
input blocks for each channel-error-to-quantization-noise ratio (CEQNR),
defined as $\sigma_e^2/\sigma_q^2$ and $\sigma_q^2 = \frac{\Delta^2}{12}$.
We vary the CEQNR and plot the resulting MSE.
The results are presented in Fig.~\ref{fig:syndQ} and compared
against the quantization error level in the existing lossy DSC methods.

\begin{figure}
  \centering
 \includegraphics [scale=0.75] {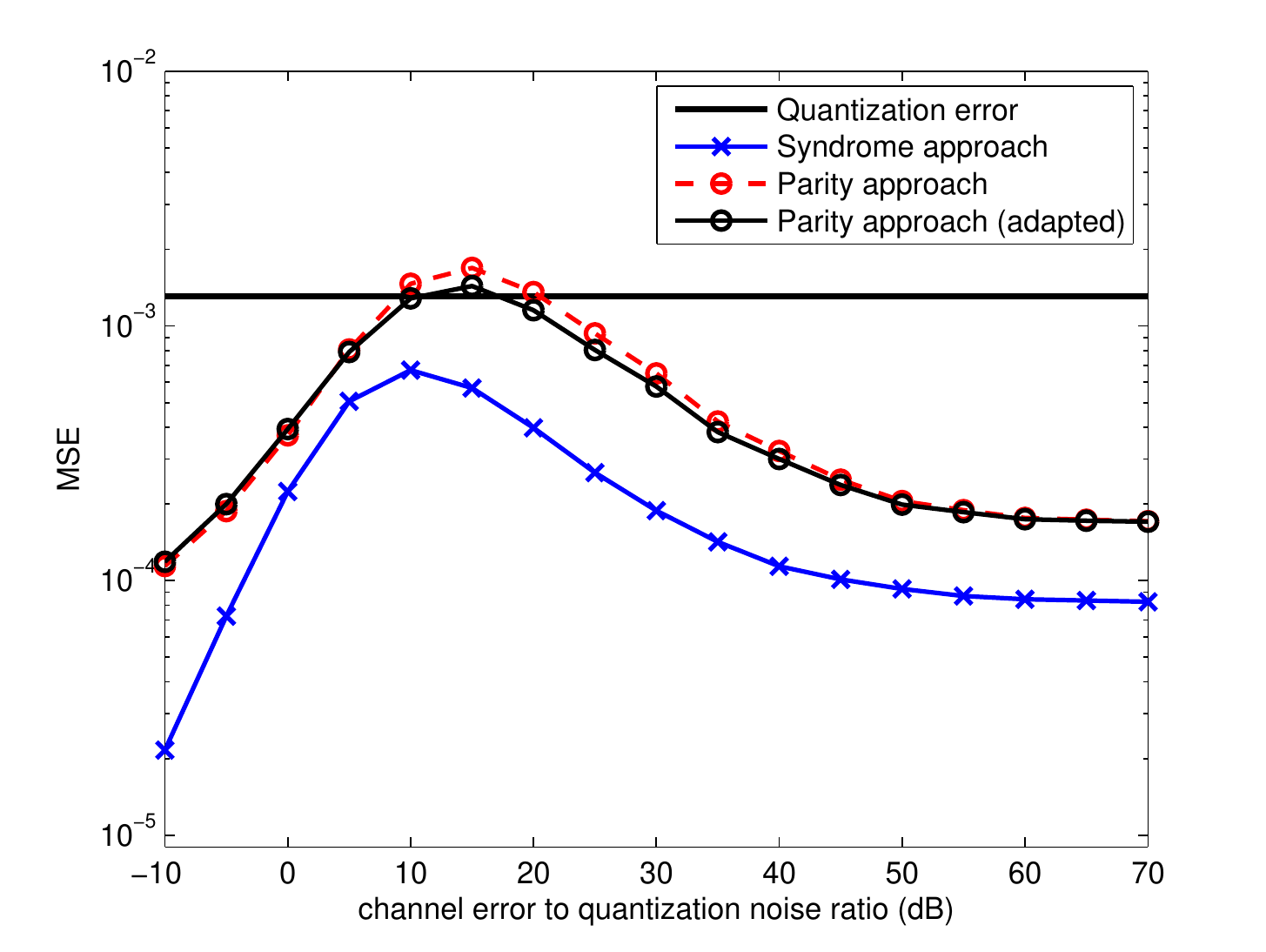}
  %\centering{\epsfig{figure=IC.eps, width=8cm}}
  \caption{ Reconstruction error in the syndrome and parity approaches, using a $(7, 5)$ DFT code in Fig. 2, 3.
  For both schemes, the virtual correlation channel inserts one error
  at each input block.}
  \label{fig:syndQ}
    %\vspace{-.05cm}
\end{figure}

It can be observed that, in the syndrome approach the MSE is lower than
the quantization error level in all CEQNRs. Analogously,
in the parity approach, the MSE is less than
that level except for a small range of CEQNR; further, it slightly improves
when we use the adapted parity approach.  Note that in lossy DSC using binary
codes, the MSE can be equal to the quantization error level only if the bit error rate (BER)
for the Slepian-Wolf coder is zero. This ideal case is not necessarily attainable even
if capacity-approaching codes are used~\cite{Vaezi2012model}.
The performance of our systems considerably improves when CEQNR is very high.
The improvement is due to a
better error localization, since the higher the CEQNR, the better the error
localization, as shown in Fig.~\ref{fig:PoE}.
%This argument is also valid for error detection, for which we use
%the algorithm proposed in Section~\ref{subsec:Pd};  for this purpose, we first find
%the threshold $\theta=0.0014$ so that $p_d=90\%$.
At very low CEQNRs, although error localization is poor, the MSE is still
very low. This is because, compared to the quantization error, the errors are so small
 that the algorithm mostly does not detect (and thus
localize) any error. Instead, it may occasionally
localize and correct quantization errors.
Note that, even if no errors are localized and corrected, the MSE is still very small
as the errors are negligible at very low CEQNRs.
Additionally, recall that the MSE is always reduced
with a factor of $R_c=\frac{k}{n}$, in an $(n,k)$ DFT code
\cite{goyal2001quantized,rath2004frame,Vaezi2011LS}.
It should be added that, for error detection, we use
the algorithm proposed in Section~\ref{subsec:Pd};  for this purpose, we first find
the threshold corresponding to $p_d=90\%$; it leads to $\theta=0.0014$.
The corresponding plots are not included, for brevity; even so,
similar plots are studied in the remainder of this section.
\begin{figure}
  \centering
  \includegraphics [scale=0.75] {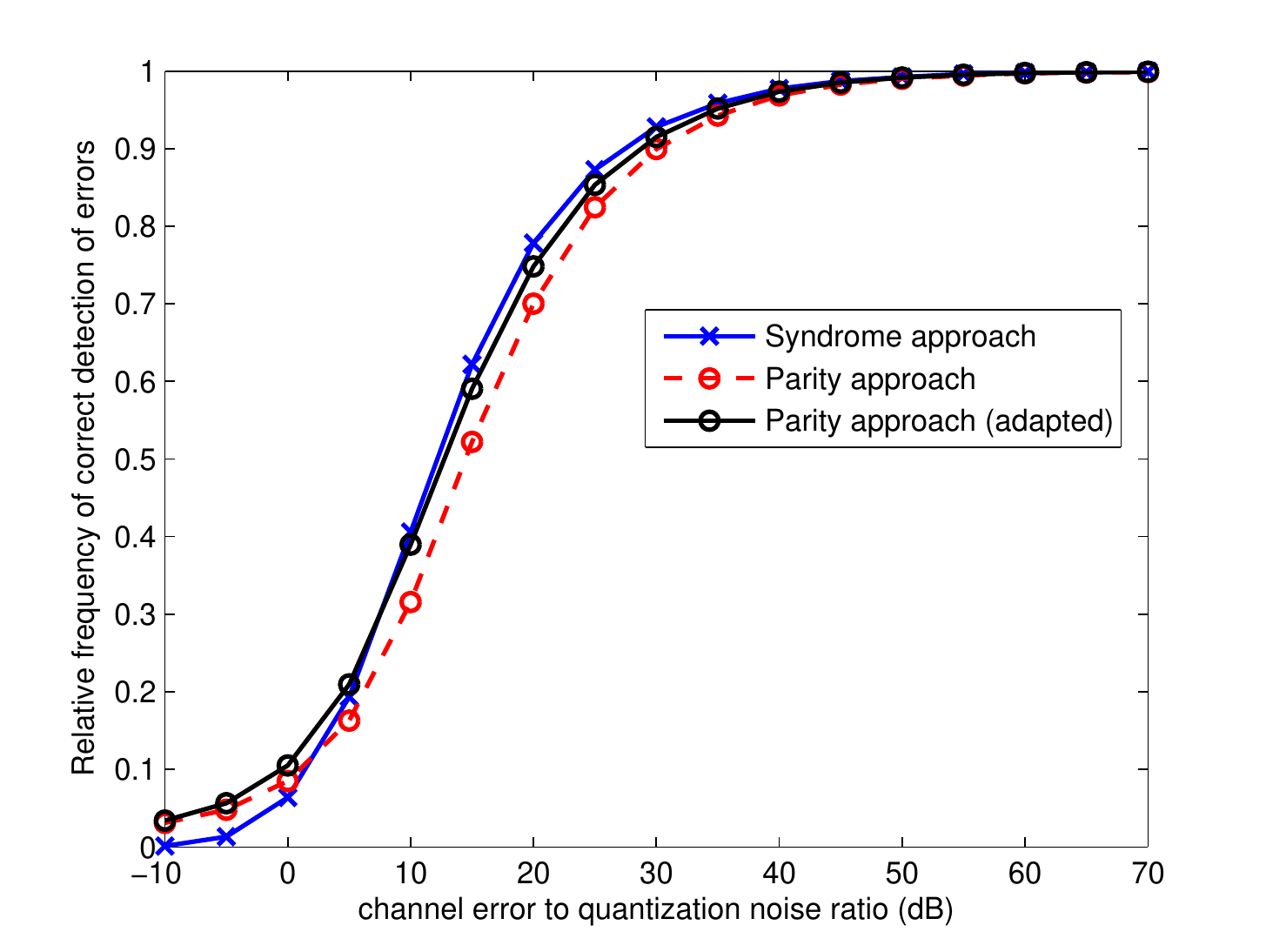}
  %\centering{\epsfig{figure=IC.eps, width=8cm}}
  \caption{ Relative frequency of correct localization of correlation channel
   error in the syndrome and parity approaches, using a $(7, 5)$ DFT code.
   Since $\nu=t=\frac{n-k}{2}$, the coding theoretic
  and subspace error localizations give the same results; hence, we employ one of them in this set of plots.
  For the parity approach, we also plot the adapted error localization,
  introduced in Section~\ref{subsec:parDec},
  as well.  }
  \label{fig:PoE}
\end{figure}

In terms of compression, the efficiency of parity approach to the syndrome approach is $R_c=\frac{5}{7}$,
 as discussed earlier in Section~\ref{comp}.
Despite this, the performance of the parity approach is
not as good as that of the syndrome approach. One reason is that in this simulation
$1/5$ of samples are corrupted in the parity approach
while this figure is $1/7$ for the syndrome approach.
More importantly,
the dynamic range of the parity samples, generated by \eqref{eq:Gsys2},
could be much higher than that of the syndrome
samples; this range depends on the position of the parity samples in the codewords.
A wider range implies more quantization levels to achieve the same accuracy.
For the $(7,5)$ code, in the best case we get  $ \sigma_p= 2.17 \sigma_x$
since $\gamma_{\min} =  2.0645$ in \eqref{eq:parVar}.
This is realized for a codeword pattern $\times \times\times -\times \times -$,
where ``$\times$'s'' and ``$-$'s'' represent data and parity samples, respectively.

\begin{figure*}[t]
\centering
\subfigure[]{
\includegraphics[scale=.53]{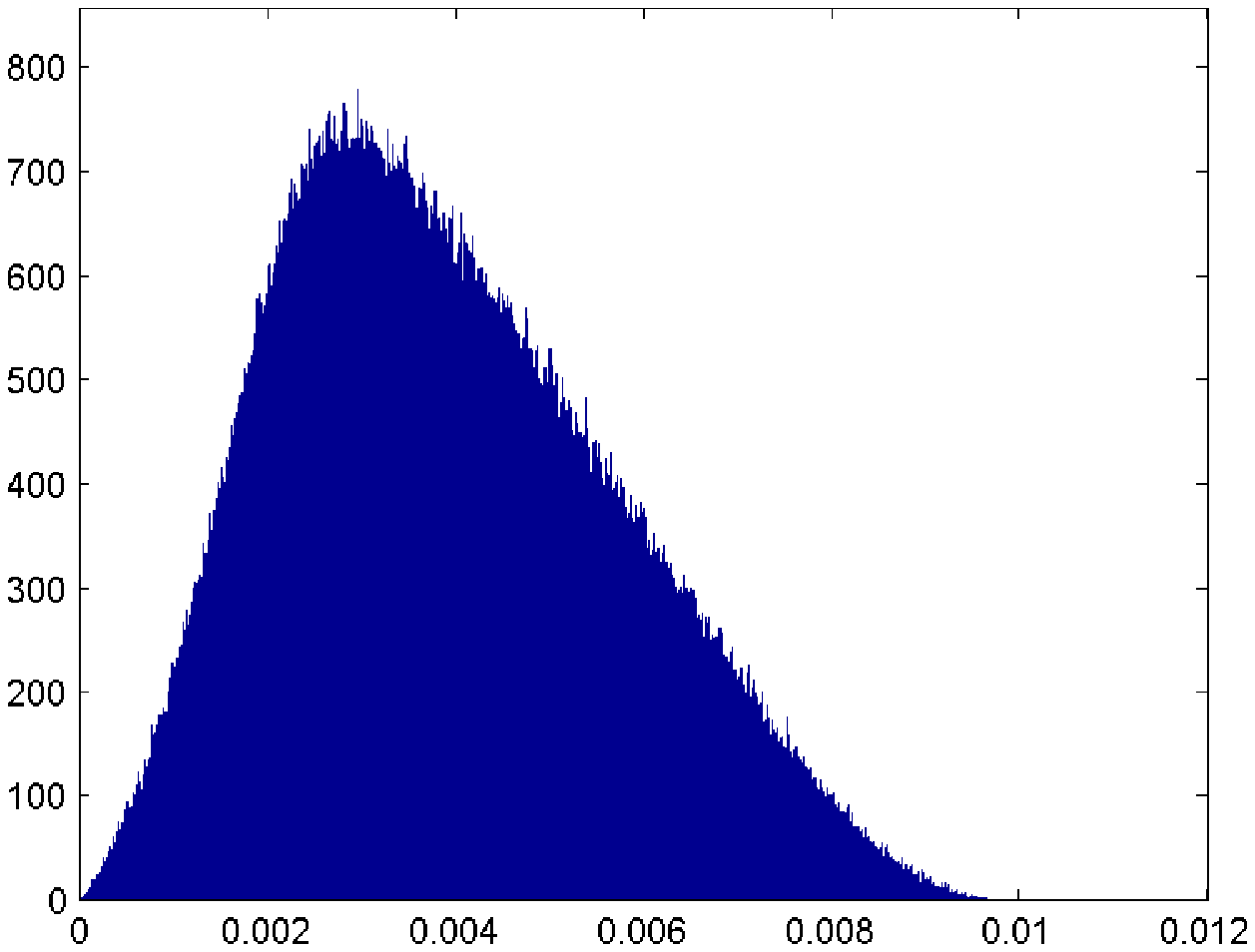}
\label{fig:subfig1}
}
\subfigure[]{
\includegraphics[scale=.53]{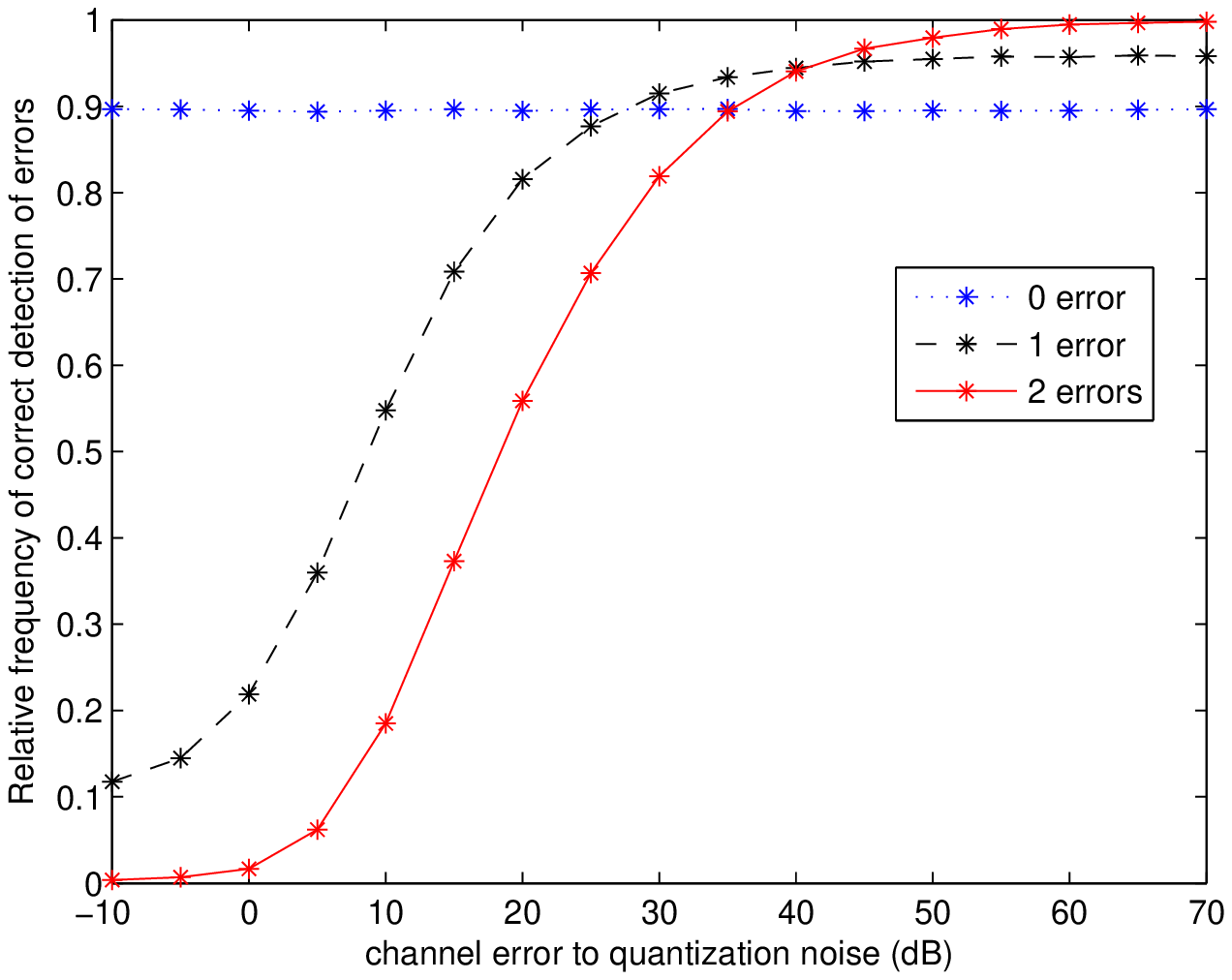}
\label{fig:subfig2}
}
\subfigure[]{
\includegraphics[scale=.53]{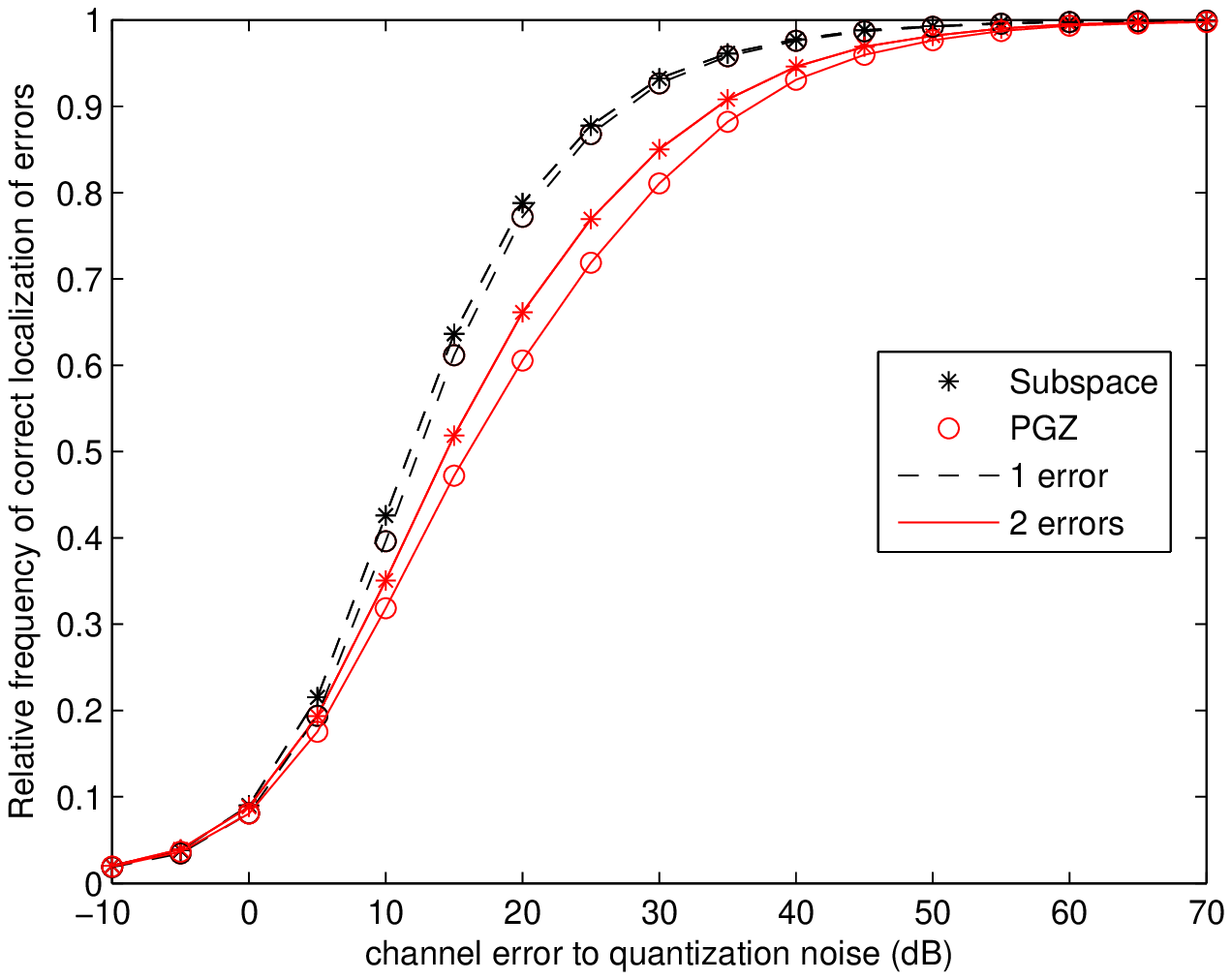}
\label{fig:subfig3}
}
\subfigure[]{
\includegraphics[scale=.53]{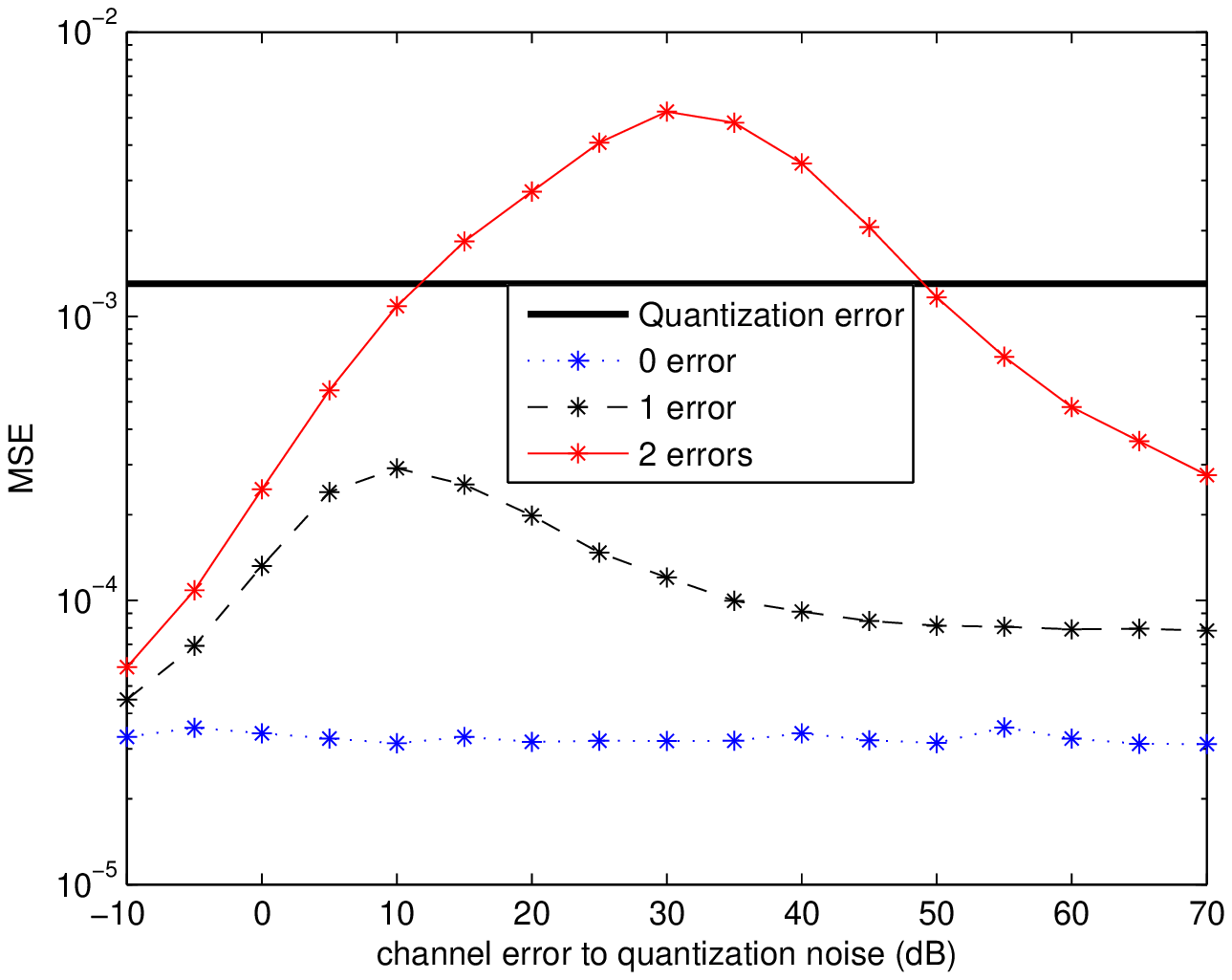}
\label{fig:subfig4}
}
\caption[Optional caption for list of figures]{Performance evaluation of joint source-channel
coding with side information at the decoder, proposed in Fig.~\ref{fig:DJSCC}, for $(10, 5)$ DFT code.
\subref{fig:subfig1} Histogram of $\lambda_{\mathrm{max}}(\tilde{\bm{R}})$ for a quantized code. This is
 used to set a threshold for detection.
\subref{fig:subfig2} Probability of correct detection of errors.
\subref{fig:subfig3} Probability of correct localization of errors.
\subref{fig:subfig4} The end-to-end distortion (MSE).}
\label{fig:simulations}
\end{figure*}

Next, we evaluate the performance of the
JSCC with side information at the decoder, illustrated in Fig.~\ref{fig:DJSCC}.
By using a systematic $(10, 5)$ DFT code, we generate, quantize, and
transmit parity samples over a noisy channel. Note that, for this code
the best systematic code \cite{vaezi2012frame} achieves the lower bound in
 \eqref{eq:parVar}; i.e., it results in $\sigma_p=\sigma_x$.
The correlation channel and transmission channel altogether insert up to $t$ errors
generated by $\mathcal N(0,\sigma_e^2)$.
Simulation results are plotted in Fig.~\ref{fig:simulations}.
First, based on Fig.~\ref{fig:subfig1}, the threshold $\theta=0.0065$ is fixed for $p_d=90\%$.
Next, this is used to estimate $\nu$ in Fig.~\ref{fig:subfig2}.
The estimated number is subsequently used to find the location of errors,
both for the PGZ and subspace methods, in Fig.~\ref{fig:subfig3}.\footnote{It is worth mentioning
that if the amplitude of errors is fixed, as assumed in \cite{rath2004subspace},
the results improve considerably for both methods. For one thing, at the CEQNR of 20dB
the probability of correct localization becomes 1.}
Then, the output of Fig.~\ref{fig:subfig3}, for the subspace method,
is fed to the last step to find the magnitude of errors and correct them.
Finally, in Fig.~\ref{fig:subfig4}, the MSE is compared against the quantization error level,
the ideal case in the lossy source coding based on
binary codes.
To put our results in perspective, we also calculate the MSE assuming perfect error localization;
it gives $0$, $6.5\times10^{-5}$, and $1.8\times10^{-4}$
 respectively for 0, 1, and 2 errors, in every CEQNR.
% This is to evaluate the effect of error localization on the end-to-end distortion.
 This implies that there is still room to improve the MSE performance
 of the proposed system, given a better solution for the error localization.
 It also shows the performance gap between this DFT code and binary codes in the ideal case.

%In \cite{zamir2002nested}, it is shown that in the quadratic
%Gaussian case the Wyner-Ziv rate-distortion function is asymptotically achievable
%as the dimensionality goes to infinity.
%
%Seeing that we do not use the ideal Slepian-Wolf
%coding assumption ($n\rightarrow \infty$), the gap between performance of the proposed
%schemes and the Wyner-Ziv rate-distortion function is more than usual. However, 
It should be noted that capacity-approaching channel codes
may introduce significant delay if one strives to
approach the capacity of the channel with very a low probability
of error. Therefore, those are out of the question
for delay-sensitive systems. In that case,
it would be best to use channel codes of low rate and focus on
achieving a very low probability of error. The system we introduced
is a low-delay system which works well with reasonably high-rate codes. %as shown in Fig.~\ref{fig:MSE}.
%It should be emphasized that low-rate DFT codes show much better performance.
%and are capable of alleviating quantization error even for $t$ errors.
Finally,  by puncturing some parity samples, rate-adaptive schemes 
are realized for the proposed 
DJSCC and parity-based DSC.
Rate-adaptive systems are popular in the transmission of non-ergodic data, like video \cite{varodayan2006rate}.

\section{Conclusions}\label{sec:con}
We have introduced a new framework for the distributed lossy source coding, in general,
and the Wyner-Ziv coding, in particular. The idea
is to do binning before quantizing the continuous-valued signal, as opposed to the conventional
approach where binning is done after quantization.
By doing binning in the real field, the virtual correlation channel
can be modeled more accurately and the
quantization error can be compensated for when there is no error.
In the new paradigm,
Wyner-Ziv coding is realized by cascading a Slepian-Wolf encoder
with a quantizer. We employ real BCH-DFT codes to do the Slepian-Wolf in the real field.
At the decoder, by introducing both syndrome-based and parity-based systems,
we adapt the PGZ decoding algorithm of channel coding to DSC.
The extension of the parity-based Wyner-Ziv coding to joint source-channel
coding with side information at the decoder is straightforward.
This scheme directly maps short source
blocks into channel blocks, and thus it is appropriate for low-delay coding.
From simulation results, we conclude that our
systems can improve the reconstruction error even using short codes,
so they can become viable in real-world scenarios where low-delay communication
is required.

We have also adapted the subspace error localization in this context
and improved it for the parity-based scheme. This reasonably improves the error localization and
leads to a better mean-squared reconstruction error.
We should point out that, a more accurate algorithm for error localization is a key to further improve the reconstruction error.
Particularly, error localization becomes more challenging when codeword length increases,
as the roots of error locator polynomial, which are over the unit circle, get closer.

\section{Appendix}

\subsection{Proof of Lemma~\ref{lem1}}
\label{sec:app1}
\begin{proof}
The proof is straightforward; we show this for odd $k$ and leave the other case to the reader.
Since $\alpha = \lceil \frac{k+1}{2}\rceil$ and $d=n-k$, using \eqref{eq:syndsamp}, for odd $k$
we can write
 \begin{align*}
s_{d-m+1}&=\frac{1}{\sqrt{n}}\sum_{p=1}^{\nu} e_{i_p}X_p^{\frac{k-1}{2}+n-k-m+1} \\
&=\frac{1}{\sqrt{n}}\sum_{p=1}^{\nu} e_{i_p}X_p^{\frac{-k+1}{2}-m} =s_m^{\ast}.
\label{eq:syndEq}
\end{align*}
Note that $X_p^n=1$, for any $p$.
\end{proof}

%\typeout{}
%\bibliography{ThesisEx}
%\bibliographystyle{ieeetr}

\end{document}